\documentclass[10pt,conference,letterpaper]{IEEEtran}
\usepackage{amsmath}
\usepackage{amssymb}
\usepackage{amsthm}
\usepackage{amsfonts}
\usepackage{xcolor,xspace}
\usepackage{tikz}

\allowdisplaybreaks

\newtheorem{lemma}{Lemma}[section]
\newtheorem{theorem}[lemma]{Theorem}

\newcommand{\de}{d}

\newcommand{\eps}{\epsilon}
\newcommand{\N}{{\mathbb{N}}}

\newcounter{note}[section]

\title{Scheduling for Weighted Flow and Completion Times in Reconfigurable Networks}
\author{\IEEEauthorblockN{Michael Dinitz\IEEEauthorrefmark{1}\thanks{\IEEEauthorrefmark{1}Supported in part by NSF awards CCF-1909111 and CCF-1464239.}}
\IEEEauthorblockA{\textit{Department of Computer Science} \\
\textit{Johns Hopkins University}\\
Baltimore, MD, USA \\
mdinitz@cs.jhu.edu}
\and
\IEEEauthorblockN{Benjamin Moseley\IEEEauthorrefmark{2}\thanks{\IEEEauthorrefmark{1}NSF Grants CCF-1830711, CCF-1733873, CCF-1733873 and CCF-1845146, a Google Research Award, a Bosch junior faculty chair and an Infor faculty award.} }
\IEEEauthorblockA{\textit{Tepper School of Business} \\
\textit{Carnegie Mellon University}\\
Pittsburgh, PA, USA \\
moseleyb@andrew.cmu.edu}
}

\date{}

\IEEEoverridecommandlockouts

\begin{document}

\maketitle

\begin{abstract}
    New optical technologies offer the ability to reconfigure network topologies dynamically, rather than setting them once and for all.  This is true in both optical wide area networks (optical WANs) and in datacenters, despite the many differences between these two settings.  Because of these new technologies, there has been a surge of both practical and theoretical research on algorithms to take advantage of them.  In particular, Jia et al.~[INFOCOM '17] designed online scheduling algorithms for dynamically reconfigurable topologies for both the makespan and sum of completion times objectives.  In this paper, we work in the same setting but study an objective that is more meaningful in an online setting: the sum of flow times.  The flow time of a job is the total amount of time that it spends in the system, which may be considerably smaller than its completion time if it is released late.  We provide competitive algorithms for the online setting with speed augmentation, and also give a lower bound proving that speed augmentation is in fact necessary.  As a side effect of our techniques, we also improve and generalize the results of Jia et al.\ on completion times by giving an $O(1)$-competitive algorithm for arbitrary sizes and release times even when  nodes have different degree bounds, and moreover allow for the \emph{weighted} sum of completion times (or flow times).  
\end{abstract}

\begin{IEEEkeywords}
Scheduling, Reconfigurable Networks
\end{IEEEkeywords}

\section{Introduction} \label{sec:intro}

The ever-increasing demand for communication has resulted in unprecedented need for data transfer in essentially all settings, from local datacenters to planetary-scale WANs.  A central challenge for network operators is to accommodate as much traffic as possible and to finish data transfers as quickly as possible.  In order to make networks even more efficient, new technologies have been developed that allow for \emph{software-reconfigurable networks} (usually just called \emph{reconfigurable networks}).  These technologies essentially allow software control over the network topology, rather than just over traditional control problems such as routing, scheduling, congestion control, etc.  In other words, we are now able to dynamically reconfigure the network topology to respond to network demands in an online fashion.

There has been a significant amount of work on actually building these technologies and systems: see~\cite{owan} for such a system for optical WANs, and see \cite{projector,firefly,mirrormirror,flyways,helios} for a small sample of reconfigurable datacenter networks (a survey of reconfigurable datacenters can be found at~\cite{survey}).  However, there has been less attention paid to the algorithmic problems raised by these technologies: we have the ability to dynamically reconfigure the network topology, but what should we reconfigure it \emph{to}?  How should we react to changing transfer and traffic demands?  Most systems use a variety of heuristics, ranging from matching-based algorithms (maximum or stable) to simulated annealing.

The theoretical study of the algorithmic challenges arising from reconfigurable networks (particularly optical WANs) was recently initiated by Jia et al.~\cite{JJGDG17}, and this remains the state of the art on the theory of scheduling reconfigurable networks.  In their setting, they assume a centralized controller that can dynamically reconfigure the network topology, with the only restriction being a degree constraint at every node (which could be different for different nodes, depending on the underlying machine represented by the node)\footnote{Clearly this is not a fully realistic setting, as in optical WANs there are optical restrictions on the topology which need to be accounted for and in the datacenter setting there is still an underlying fixed network in addition to some reconfigurable links.  But as discussed in~\cite{JJGDG17}, it is a reasonable starting point for developing algorithms.}.  There is a stream of transfer requests arriving at the system, where each request has a source, a destination, a transfer size, and a release time (the earliest time by which the transfer can start).  The goal is to design a scheduling algorithm that decides, at each time slot, what topology to build and what jobs to transfer using that topology (under the additional restriction that multihop paths are not allowed).  They provided both offline and (more interestingly) online algorithm for the makespan objective (minimizing the time at which all transfers are finished) and the sum of completion times (minimizing the sum over all jobs of the time at which they finished).

We work in the same model, but extend and improve the results of~\cite{JJGDG17}.  Most importantly, we provide  online algorithms and prove their competitive ratio for a more natural objective function: the (weighted) sum of \emph{flow times}.  The flow time of a job (also sometimes called the \emph{sojourn time}, \emph{waiting time}, or \emph{response time}) is simply the time that it is in the system, i.e., its completion time minus its release time. If all release times are $0$, then flow times and completion times are the same.  But if jobs are released online, then not only are they extremely different, but moreover approximation guarantees on the completion times are not particularly meaningful.  While both problems have the same optimal solution, in an approximation analysis one can make a job wait proportional to its release date with little penalty.  When the time horizon is large, then  undesirable schedulers can have a small (e.g. constant) approximation ratio.

\medskip
\noindent \textbf{Completion Time Versus Flow Time.} To see the difference between completion time and flow time consider an extremely simple example, suppose that there are only two jobs, each of which has size $1$.  Job $1$ is released at time $1$, and job $2$ is released at time $1000$.  Then consider the schedule which which schedules job $1$ at time $999$ and job $2$ at time $1000$.  Clearly this is an undesirable schedule -- we should have scheduled job $1$ at time $1$ and job $2$ at time $1000$.  But if we look at the sum of completion times, the optimal solution has cost $1001$, while this horrible schedule has cost $1999$.  So this horrible schedule looks pretty good with respect to completion times, since it is a $2$-approximation!  This is clearly ridiculous; we ``cheated" by allowing job $1$ to have terrible performance but  it  balanced out with job $2$'s release date.  On the other hand, if we look at the sum of flow times, the horrible schedule has cost $1000$ (since job $1$ is in the system for $999$ time units while job $2$ is only in the system for $1$ time unit) while the optimal solution has cost $2$ (since jobs do not have to wait to be scheduled).  Thus, the flow time objective will rule out such a schedule and   accurately reflects the quality of a schedule. 

\medskip
\noindent \textbf{Results:} In this paper we initiate the study of reconfigurable network scheduling under the weighted flow time objective.  In more detail, we prove the following results.

\begin{itemize}
\item In the most general setting of~\cite{JJGDG17} (arbitrary degree constraints, arbitrary job sizes, arbitrary release dates), and in addition where every job has a weight which multiplies its flow time in the objective, we give an algorithm which is $O(1/\epsilon^2)$-competitive as long as the algorithm is allowed to have \emph{speed} $2+\epsilon$, for any $\epsilon > 0$.  Informally, this is a form of resource augmentation: we allow the algorithm to complete jobs at a rate that is $2+\epsilon$ faster than the optimal solution is allowed.  From a networking perspective, this is equivalent to allowing higher throughput edges as resource augmentation.  So, for example, our algorithm will have weighted flow time with $(2+\epsilon)100$ Gbps links that is only $O(1/\epsilon)$ times worse than the optimal solution with $100$ Gbps links.  This can also be thought of as overprovisioning: if we want performance that is comparable to the optimum but without knowing in advance what the jobs will look like, then we can just overprovision by a $2+\epsilon$ factor. 


\item We justify our previous results by showing that speed augmentation is necessary: we prove a \emph{polynomial} lower bound on any online algorithm without speed augmentation.  In particular, we prove that any online randomized algorithm without speed augmentation can have competitive ratio that is at best $\Omega(\sqrt{n})$.  This is a terrible lower bound, showing that without resource augmentation all algorithms perform poorly in the worst case.  In settings like this, resource augmentation has been used so theory can differentiate between the performance of algorithms \cite{PruhsST04}.

\item As a side effect of our techniques, we are also able to extend the results of~\cite{JJGDG17} on completion times to a more general setting.  While this work provided many algorithms and $O(1)$-competitive analyses, they did not give an $O(1)$-competitive algorithm for the most general case: general degree constraints, general job sizes, and general release times.  They also did not give bounds on \emph{weighted} completion times.  A simple modification of our flow time algorithm gives an $O(1)$ approximation \emph{without} speed augmentation for the completion time objective in the most general setting. 
\end{itemize}

\noindent \textbf{Outline.} In Section~\ref{sec:related} we describe related work for both reconfigurable networks and flow time scheduling in other settings.  In Section~\ref{sec:preliminaries} we formally describe the problem setting.  Section~\ref{sec:upper} has our main upper bounds.  We begin with a warm-up in Section~\ref{sec:simple} where we assume that all weights are $1$, all job sizes are $1$, and all degree bounds are the same.  This simplified setting allows us to demonstrate the intuition behind our more general techniques.  We then prove give our algorithm and analysis for the general setting in Section~\ref{sec:general}, and show how this can be modified to give a bound on completion times in Section~\ref{sec:completion}.  Finally, in Section~\ref{sec:lower-bound} we prove our lower bound implying that speed augmentation is necessary.

\section{Related Work} \label{sec:related}

\subsection{Reconfigurable Networks}
As discussed in the introduction, there has been a significant amount of work in the last decade on reconfigurable datacenters.  For overviews, see a recent tutorial from SIGMETRICS 2019~\cite{reconfigurable-tutorial} and the related survey on reconfigurable datacenters~\cite{survey}.  These have been enabled by a variety of technologies, including optical circuit switching~\cite{helios,cThrough}, 60GHz wireless~\cite{flyways}, and free space optics~\cite{firefly,projector}.  

From an algorithmic point of view, these systems generally use a variety of heuristics without provable guarantees.  The main line of work on understanding the theory behind reconfigurable datacenters is in the form of \emph{demand-aware networks}~\cite{DAN1,DAN2,DAN3,DAN4}.  In this setting, we assume that we are given a traffic matrix, and are trying to design a network topology which will have good performance on that traffic matrix (i.e., since the network is reconfigurable we can measure demand and then build an appropriate network topology).  Usually the notion of quality involves the (average) lengths of paths.  Scheduling problems are not considered in this setting.

For non-datacenter contexts, reconfigurable optical WANs were introduced by~\cite{owan}.  The scheduling algorithms used in~\cite{owan} were based on heuristics (simulated annealing in particular), so in followup work, Jia et al.~\cite{JJGDG17} introduced the theoretical study of scheduling algorithms for reconfigurable optical WANs.  They worked in a model which is not a perfect match for optical WANs, but is close enough to be useful.  We adopt this model, and extend~\cite{JJGDG17} to a better objective function and slightly more general setting.  Moreover, since their model ignores many of the real-world difficulties of optical WANs, it applies to more general reconfigurable networking settings.  

We note that while WANs and datacenters are obviously extremely different settings, our goal is to understand the scheduling problems that arise from the power of reconfiguration.  Hence we abstract out to a level which encompasses both of these settings, at the price of not being extremely realistic for either of them.  However, this is the level of abstraction used in~\cite{JJGDG17}, so it is perhaps a reasonable setting for optical WANs.  For datacenters, the main difference between our model and reality is the existence of an underlying fixed network: in our model we assume that the entire network is reconfigurable, while in most reconfigurable datacenter systems only a fraction of the links can be reconfigured.  Analyzing this combined setting is an interesting future line of research, which was recently initiated in the context of routing~\cite{FGS18,FPS19} but which is still entirely unexplored for scheduling.


\subsection{Flow Time Scheduling}

Optimizing total weighted flow time is the most popular objective in online scheduling theory.  We discuss related work on the problem of scheduling $n$ preemptive jobs that arrive over time on a single machine with the objective of optimizing the total weighted flow time.   For a (slightly dated) survey see~\cite{PruhsST04}, and further pointers to relevant work can be found in~\cite{ImMP11}.  It is folklore that the algorithm Shortest-Remaining-Processing-Time (SRPT) is optimal for scheduling unweighted jobs on a single machine.  When jobs have weights, it is known that no online algorithm can have a constant competitive ratio \cite{BansalC09}.  

When there are non-constant lower bounds the competitive ratio of any online algorithm, prior work has focused on a resource augmentation analysis.  A $s$-speed $c$-competitive algorithm is one where the algorithm achieves a competitive ratio of $c$ and the algorithm is given a machine that is a factor $s$ faster than the optimal solution.  The consensus in the community is that the best positive theoretical result one can show is an algorithm that is $(1+\epsilon)$-speed $f(\epsilon)$-competitive for any constant  $\epsilon >0$ where $f(\cdot)$ is a function only depending on $\epsilon$ \cite{KalyanasundaramP00}.   In particular, the competitive ratio is independent of $n$, e.g., $O(\frac{1}{\eps})$.  Such an algorithm is known as \emph{scalable}. Showing an algorithm is scalable gives strong evidence that the algorithm will work well in practice. 

 The most natural algorithm is highest-density-first when jobs have weights.  This algorithm prioritizes jobs in order of their weight over processing time.  This algorithm is known to be  $(1+\epsilon)$-speed $O(\frac{1}{\eps})$-competitive for total weighted flow time on a single machine \cite{BecchettiLMP06}.  The algorithm has been generalized to many environments \cite{AnandGK12}.
\section{Definitions and Preliminaries} \label{sec:preliminaries}

As discussed, we will be studying the same model as~\cite{JJGDG17}.  The main difference is the objective function.  

\paragraph{Model and Scheduling Definition} There is a set of nodes $V$, each representing a node in our network.  Each vertex $v \in V$ comes with a degree bound $d_v$.  A request (job) is a tuple $(u_i, v_i, \ell_i, r_i, w_i)$, where $u_i, v_i \in V$ are the source and destination respectively, $\ell_i \in \N$ is the size, $r_i \in \N$ is the release time, and $w_i \in \mathbb{R}$ is the weight.  Note that without loss of generality we assume sizes and release times are natural numbers, since we can always adjust the scale of a time slot.  In each round $t$, we can create a graph $G_t$ with vertex set $V$ which satisfies the degree constraints, and where each edge $\{u,v\} \in E(G_t)$ is labeled with a request $i$ such that $\{u,v\} = \{u_i, v_i\}$ and $t \geq r_i$.  The request is completed once it has appeared in at least $\ell_i$ of these graphs.  Note that as in~\cite{JJGDG17} we are allowing only direct links (we do not allow data to be transferred over longer paths) and allow preemption.  See~\cite{JJGDG17} for more justification of this model.   

\paragraph{Online vs Offline} Clearly scheduling problems in this context make sense both on- and offline.  We will be concerned with the competitive ratio (the worst case cost of the algorithm divided by the optimal solution) of scheduling in the online setting.  This the same as the approximation ratio, except we require the algorithm to be online.

\paragraph{Objective Function and Speed Augmentation} As discussed, Jia et al.~\cite{JJGDG17} considered two objective functions: the makespan and the sum of completion times. 
We will mostly be concerned with a different measure of quality: the weighted sum of \emph{flow times}.  The flow time of a request $i$ is the time $c(i)$ at which it completes minus its release time $r_i$.  That is, the flow time of a job is simply how long it is in the system before being completed.  This is a more natural objective than the sum of completion times, but is also more difficulty to optimize.  We will consider the objective of the weighted flow time, where our goal is to minimize $\sum_{i} w_i(c(i) - r_i)$.

Unfortunately, as we show in Section~\ref{sec:lower-bound}, it is not possible to provide $O(1)$-competitive algorithm for the total flow time, even when all weights and sizes are unit.  In the face of strong lower bounds we adopt the most popular form of analysis known as a resource augmentation analysis.  Here we give the algorithm extra speed.  An algorithm running with speed $s \geq 1$ is able to process jobs at a rate that is $s$ times faster than the optimal solution.  As discussed in Section~\ref{sec:intro}, this can be thought of as overprovisioning the network, and will allow us to design competitive algorithms for the flow time objective.  Moreover, as discussed in Section~\ref{sec:related}, this notion of speedup is relatively standard in the scheduling literature.

\section{Upper Bounds} \label{sec:upper}
In this section we give our algorithms and corresponding upper bound results.  We begin in Section~\ref{sec:simple} with a simple setting that serves to demonstrate most of the main ideas behind our algorithm and analysis.  In Section~\ref{sec:general} we move to the most general online setting to prove our main results.  

\subsection{Simple setting} \label{sec:simple}
We will begin with the simplest possible setting: when all degree bounds are equal to $1$, all job sizes are $1$, and all weights are $1$.  Note that, in particular, since all degree bounds are $1$ the set of jobs scheduled at any time form a matching.

\subsubsection{Algorithm}
At time $t$, let $G(t)$ be the (multi)graph of all jobs that are in the system at time $t$ (i.e., all requests with release times at most $t$ which have not already been completed).  Order the jobs by release time (breaking ties arbitrarily but consistently), and then construct a maximal matching $E_t$ using this ordering.  These are the jobs scheduled at time $t$.  For each job $i$, let $C_i$ be the completion time of job $i$ (the time at which it is scheduled by this algorithm).

\subsubsection{Analysis}
While the algorithm itself is simple and combinatorial, we will analyze it through an LP relaxation, and in particular through the technique of dual fitting.  Let $\mathcal S$ denote the set of all jobs.   Consider the following linear program.

\begin{alignat*}{3}
\min & \quad  \sum_{i \in \mathcal S} \sum_{t \geq r_i} (t-r_i) x_{i,t} \\
\text{s.t.} & \quad \sum_{t \geq r_i} x_{i,t} \geq 1& \quad\forall i \in \mathcal S \\
& \quad \sum_{i \in \mathcal S : |\{u_i, v_i\} \cap \{w\}| = 1} x_{i,t} \leq 1  & \quad \forall w \in V,\ \forall t \in \N \\
& \quad x_{i,t} \geq 0 & \quad \forall i \in \mathcal S,\ \forall t \in \N
\end{alignat*}

While technically this LP has infinite size (since we did not put an upper bound on $t$), it is easy to see that we can put an upper bound on $t$ of $n \cdot \max_{i \in \mathcal S} r_i$, so this LP has finite size.  It is easy to show that this is a feasible LP relaxation.

\begin{lemma}
    If there is a schedule with sum of flow times at most $F$, then there is a solution to the LP of cost at most $F$.
\end{lemma}
\begin{proof}
Consider a schedule $\{E_t\}_{t \in \mathbb{N}}$ with sum of flow times $F$.  Since this is a feasible schedule, each $E_t$ is a matching.  We create an LP solution as follows: if job $i$ is scheduled at time $t$, then we set $x_{i,t} = 1$, otherwise we set $x_{i,t} = 0$.  Since the original schedule is feasible, every job is scheduled in some $t$ so the first LP constraint is satisfied, and similarly since each $E_t$ is a matching the second LP constraint is satisfied.  Thus this is a feasible LP solution.  By the definition of the $x$'s, the flow time in the schedule is precisely $\sum_{t \geq r_i} (t - r_i) x_{i,t}$, and thus the LP objective is the sum of the flow times, $F$.
\end{proof}

The dual of this LP is the following.

\begin{alignat*}{3}
\max & \quad  \sum_{i \in \mathcal S} \alpha_i - \sum_{u \in V} \sum_{t \in \N} \beta_{u,t} \\
\text{s.t.} & \quad \alpha_i - \beta_{u_i, t} - \beta_{v_i, t} \leq t - r_i& \quad\forall i \in \mathcal S,\ \forall t \in \N \\
& \quad \alpha_i \geq 0& \quad \forall u \in \mathcal S \\
& \quad \beta_{i,t} \geq 0& \quad \forall i \in \mathcal S,\ \forall t \in \N
\end{alignat*}

We will analyze our algorithm by finding a feasible dual solution and relating this to the cost of the algorithm.  However, due to the lower bound in Section~\ref{sec:lower-bound}, we will need to allow \emph{resource augmentation}.  Let $ALG(s)$ denote the total flow time of the algorithm when run with speedup $s$, i.e., when the algorithm processes jobs at a speed of $s$. 

Let's now define our dual solution.  But first we need a little bit of notation: for every node $v \in V$ and time $t$, let $d_v(t)$ denote the degree of $v$ in $G(t)$.  Then for every $i \in \mathcal S$, we let $\alpha_i = \frac{d_{u_i}(r_i) + d_{v_i}(r_i)}{2s}$.  Similarly, we will set $\beta_{u,t} = d_u(t) / (2s)$.

We first show that this is a feasible dual solution.

\begin{lemma} \label{lem:feasible}
    $\alpha_i - \beta_{u_i, t} - \beta_{v_i, t} \leq t - r_i$ for all $i \in \mathcal S$ and $t \geq r_i$.
\end{lemma}
\begin{proof}
We prove this by induction on $t$.  For the base case, let $t = r_i$. Then
\begin{align*}
    \alpha_i - \beta_{u_i, t} - \beta_{v_i, t} &= \frac{d_{u_i}(r_i) + d_{v_i}(r_i)}{2s} - \frac{d_{u_i}(r_i)}{2s} - \frac{d_{v_i}(r_i)}{2s} \\
    &= 0 = t - r_i,
\end{align*}
as claimed.  Now consider some $t > r_i$.  Note that since we allow speedup $s$, the number of jobs scheduled at one time that have some fixed node as an endpoint is at most $s$ (rather than at most $1$).  Thus
\begin{align*}
    &\alpha_i - \beta_{u_i, t} - \beta_{v_i, t}= \alpha_i - \frac{d_{u_i}(t)}{2s} - \frac{d_{v_i}(t)}{2s}  \\
    &\leq \alpha_i - \frac{d_{u_i}(t-1) - s}{2s} - \frac{d_{v_i}(t-1) - s}{2s} \\
    &= \alpha_i - \beta_{u_i, t-1} + \frac12 - \beta_{v_i, t-1} + \frac12\\
    &\leq (t-1-r_i) + 1 = t-r_i. \qedhere
\end{align*}
\end{proof}

We will now prove two lemmas which will allow us to bound the cost of this dual solution.

\begin{lemma} \label{lem:alpha}
    $\sum_{i \in \mathcal S} \alpha_i \geq \frac12 \cdot ALG(s)$.
\end{lemma}
\begin{proof}

We first claim that in the algorithms (with speedup $s$), the flow time of job $i$ is at most $\frac{d_{u_i}(r_i) + d_{v_i}(r_i)}{s}$.  To see this, let $\mathcal S_i$ be the set of jobs $j$ with  $r_j < r_i$ and $\{u_j, v_j\} \cap \{u_i, v_i\} \neq \emptyset$ that have not been completed by time $r_i$.  Note that $|\mathcal S_i| = d_{u_i}(r_i) + d_{v_i}(r_i)$ by definition.  Now consider some time $t$ after job $i$ has been released. If job $i$ has not yet been completed, and is not scheduled at time $t$, then some job $j \in \mathcal S_i$ must be scheduled at time $t$.  This is because the algorithm sorts by release time and constructs a greedy maximal matching in this order. In particular, if no  job $j$  in $S_i$ is scheduled at time $t$, then  we will schedule job $i$.  Thus the time that $i$ spends in the system before being scheduled is at most $d_{u_i}(r_i) + d_{v_i}(r_i)$.  Since we have speedup $s$, the flow time of job $i$ is at most $\frac{d_{u_i}(r_i) + d_{v_i}(r_i)}{s}$.

This now allows us to analyze the $\alpha$ variables.  We get that
\begin{align*}
    \sum_{i \in \mathcal S} \alpha_i &= \sum_{i \in \mathcal S} \frac{d_{u_i}(r_i) + d_{v_i}(r_i)}{2s} = \frac12 \sum_{i \in \mathcal S} \frac{d_{u_i}(r_i) + d_{v_i}(r_i)}{s} \\
    &\geq \frac12 \sum_{i \in \mathcal S} \left(C_i - r_i\right) = \frac12 \cdot ALG(s). \qedhere
\end{align*}
\end{proof}

\begin{lemma} \label{lem:beta}
    $\sum_{w \in V} \sum_{t \in \N} \beta_{w,t} \leq \frac{1}{s} \cdot ALG(s)$.
\end{lemma}
\begin{proof}
This is essentially a straightforward calculation using the fact that the flow time of a job is equal (by definition) to the number of time steps in which the job is in the system.  So we have that
\begin{align*}
    \sum_{w \in V} \sum_{t \in \N} \beta_{w,t} &= \frac{1}{2s} \sum_{t \in \N} \sum_{w \in V} d_u(t) = \frac{1}{s} \sum_{t \in N} |E(G(t))|\\
    &= \frac{1}{s} \sum_{i \in \mathcal{S}} (C_i - r_i) = \frac{1}{s} \cdot ALG(s),
\end{align*}
as claimed.
\end{proof}

We can now prove our main theorem (about this simple setting).

\begin{theorem}
$ALG(2+\epsilon) \leq \frac{2(2+\epsilon)}{\epsilon} \cdot OPT$ for any $\epsilon > 0$.
\end{theorem}
\begin{proof}
Let $s = 2+\epsilon$.  Combining Lemmas~\ref{lem:alpha} and \ref{lem:beta} implies that 
\begin{small}
\begin{align*}
    \sum_{i \in \mathcal S} \alpha_i - \sum_{w \in V} \sum_{t \in \N} \beta_{w,t} &\geq \frac12 \cdot ALG(2+\epsilon)  - \frac{1}{2+\epsilon} \cdot ALG(2+\epsilon) \\
    &= \frac{\epsilon}{2(2+\epsilon)} \cdot ALG(2+\epsilon).
\end{align*}
\end{small}

We know from Lemma~\ref{lem:feasible} that $(\alpha, \beta)$ is a feasible dual solution, so by weak duality we get that
\begin{align*}
    ALG(2+\epsilon) &\leq \frac{2(2+\epsilon)}{\epsilon} \cdot \left(\sum_{i \in \mathcal S} \alpha_i - \sum_{w \in V} \sum_{t \in \N} \beta_{w,t} \right) \\\
    &\leq \frac{2(2+\epsilon)}{\epsilon} \cdot OPT. \qedhere
\end{align*}
\end{proof}

\subsection{General Online Model} \label{sec:general}

This section considers the most general model.   In this case each node $v$ has a degree bound $\de_v$ denoting the maximum number of jobs involving $v$ that can be scheduled at any point in time. A job $i$ has size $\ell_i$  and a weight $w_i$.    We will assume there is no restriction on how much a job is scheduled, so long as the degree constraints are satisfied at the vertices.  We will let $h_i = \frac{w_i}{\ell_i}$ be the \emph{density} of job $i$.  The goal is to optimize the total weighted flow time $\sum_{i \in [n]} w_i (c(i) - r_i)$. 

This section is organized as follows.  We first give our algorithm, which is simple and natural (highest-density-first).  We then spend most of the section analyzing it.  To do this, we show that we can focus on a different objective called weighted \emph{fractional} flow time. We will call the original objective weighted \emph{integral} flow to differentiate them.  We show that if the algorithm performs well for the fractional objective then the algorithm performs well for the integral objective with slightly more speed up.   Once we focus on the fractional flow objective, we can further show that we may assume all jobs are unit time in the analysis after scaling the weights.  We note that both of these reductions are done to simplify the analysis -- the algorithm itself does not change or make any of these assumptions, and could be analyzed directly (although doing so is more technical and complicated). 

With these simplifications and reductions in place we perform a dual-fitting analysis of the algorithm.  As in the simple case of Section~\ref{sec:simple}, the intuition is that the dual variables correspond to the ``extra cost'' to the algorithm incurred by a job when it arrives.  This is more complicated than in the simple setting due to the addition of weights and job size (or just weights after the reductions), but the ideas are the same.

\subsubsection{Algorithm: Highest-Density-First}

Recall that $\mathcal S(t)$ is the set of released but uncompleted jobs at time $t$.   When scheduling, we  say a node $u$ is saturated if it schedules $d_u$ jobs adjacent to it.  Order the jobs in  $\mathcal S(t)$ in decreasing order of their density.   In this order, schedule job $i$ if the two endpoints for $i$ are not saturated. We note that we schedule job $i$ as must as possible if its endpoints are not saturated, that is, we will create parallel links between the endpoints until one of them is saturated or the job is completely scheduled. 

\subsubsection{Reduction to the Unit Time Case}

This section is devoted to proving the following lemma, stating that we may assume in the analysis that each job is restricted to only being unit size but arbitrary weight.  This transformation is done only to simplify the analysis; the algorithm itself is unaffected.  

\begin{lemma}\label{lem:unit}
If highest-density-first is $s$-speed $c$-competitive on unit size instances, then highest-density-first is $(1+\eps)s$-speed $\frac{(1+\eps) c}{\epsilon}$-competitive for arbitrary size and arbitrary weight instances.
\end{lemma}

To prove the lemma first consider a different objective called \textbf{weighted fractional flow time}.  To make the distinction between these objectives, we call the original objective \textbf{weighted integral flow time}.  
Recall that $\mathcal S(t)$ is the released but uncompleted jobs at time $t$.  For each job $i \in \mathcal S(t)$ let $0 \leq \ell_i(t) \leq \ell_i$ be the remaining size of job $i$ at time $t$.   Then we define the weighted fractional flow time to be $\sum_{t \in \mathbb{N}} \sum_{i \in \mathcal S(t)} w_i \frac{\ell_i(t)}{\ell_i} $.  In this objective, each job $i$ pays $w_i \frac{\ell_i(t)}{\ell_i} $ at each time $t$ it is alive and unsatisfied.  Note that the original weighted integral flow time objective is equivalent to $\sum_{t \in \mathbb{N}} \sum_{i \in \mathcal S(t)} w_i$, and hence the difference between the two objectives is that in the fractional objective the weight of a job is scaled by $\frac{\ell_i(t)}{\ell_i}$ (the fraction of the job size that is uncompleted).

We now show that we can convert any algorithm for fractional flow to one for integral flow time (and thus in particular the highest-density-first algorithm). 

\begin{lemma}\label{lem:convert}
Given any online algorithm $A$ with $s$-speed that is $c$-competitive for fractional flow time, for any $\epsilon > 0$ there is an online algorithm $B$ that is $(1+\epsilon)s$-speed $\frac{(1+\eps) c}{\epsilon}$-competitive for integral flow time.  Further if $A$ is highest-density-first, so is $B$. 
\end{lemma} 
\begin{proof}
Consider the algorithm $A$ for fractional flow time.  Each time $A$ schedules a job $i$ with speed $s$ the algorithm $B$ processes the same job with speed $(1+\eps)s$. If the job has already been completed in $B$ then $B$ can either be idle or work on some other job (e.g., the remaining with highest density).  Clearly the schedule produced by algorithm $B$ is feasible if the schedule produced by algorithm $A$ is feasible, since no job is scheduled by $B$ before it is released.   Notice that if $A$ is highest-density-first then $B$ can be highest-density-first. This is because highest-density-first has the property that if the algorithm is given more speed then the algorithm will either process the same job as the slower schedule or the algorithm will have completed the job.

Fix any job $i$. Consider the first time $t_i$ where a $\frac{1}{1+\eps}$ fraction of $i$ is completed in $A$.  So $\frac{\ell_i(t)}{\ell_i} \geq\frac{\eps}{1+\eps}$ for all $t \leq t_i$.  Thus every $t$ with $r_i \leq t \leq t_i$ contributes $\frac{\eps w_i}{1+\eps}$ or more to the objective.  Since $B$ schedules job $i$ at the  same times or earlier as $A$ with speed a $(1+\eps)$ factor faster, $B$ will complete the job by time $t_i$.  So $B$ pays at most $w_i$ for each $t$ with $r_i \leq t \leq t_i$, while $A$ pays at least $\frac{\eps w_i}{1+\eps}$.  Hence the ratio between the two costs is at most $\frac{1+\eps}{\eps}$.   

This holds for all jobs.  Further, the fractional optimal objective is only less than the integral optimal objective.  This gives the lemma.
\end{proof}

The previous lemma shows that we may focus on the weighted fractional flow time objective.  The next lemma shows that we can further restrict the instance to unit size jobs.  Combining these two lemmas will allow us to focus on the unit size case. 

\begin{lemma} \label{lem:unit-fractional}
For the fractional flow time objective, any instance can be transformed to a different problem instance such that (1) the objective for the highest-density-first algorithm is the same on \emph{both} instances, (2) the optimal objective is only less on the new instance, and (3) in the transformed instance all jobs are unit size.
\end{lemma}
\begin{proof}
 Fix any instance. Consider transforming any job $i$ into $\ell_i$  new jobs $i'_1, i'_2, \ldots, i'_{\ell_i}$.  Each new  job $i'$ has size $1$ and weight $\frac{w_i}{\ell_i}$.    Note that the density of the jobs $i'_j$ are the same as $i$ for all $j$. 

Consider any schedule $A$ for the original instance. We create the analogous schedule $B$ for the new instance.  Whenever a job $i$ is processed by $A$ for $k$ units at some time $t$,  jobs $ \{i'_{j}, i'_{j+1}, \ldots i'_{j+k} \}$ are processed by $B$ such that $j$ is the lowest index possible among unsatisfied jobs. Both schedules then are intuitively working on the same job at the same times.  Notice that $A$ is highest-density-first on the original instance if and only if $B$ is the highest-density-first algorithm on the new instance, since the density of the jobs in $I'_i$ are the same as $i$. 

The fractional flow time objective is the same for $A$ and $B$  because each time $\ell_i(t)$ decreases by $1$, the weight of $i$ in $A$ changes from $w_{i} \frac{\ell_i(t)}{\ell_i}$ to $w_{i} \frac{\ell_i(t)-1}{\ell_i}$.  Similarly in $B$, there are $\ell_i(t)$ jobs alive in $I'_i$ and this decreases by $1$.  Their weight was $|I'_i| \frac{w_i}{\ell_i} = w_{i} \frac{\ell_i(t)}{\ell_i}$ and this decreases to $(|I'_i| -1) \frac{w_i}{\ell_i} = w_{i} \frac{\ell_i(t) -1}{\ell_i}$.
\end{proof}

Lemmas~\ref{lem:convert} and \ref{lem:unit-fractional} imply that if highest-density-first is $s$-speed $c$-competitive for unit-size jobs with respect to weighted fractional flow time, then for any $\epsilon > 0$, highest-density-first is $(1+\epsilon)s$-speed $\frac{(1+\epsilon)c}{\epsilon}$-competitive for general size jobs with respect to weighted integral flow time.  But for unit-size jobs, the fractional flow time is equal to the integral flow time.  Thus we have proved Lemma~\ref{lem:unit}.

\subsubsection{Analysis}


As in the simple setting of Section~\ref{sec:simple}, we perform a dual fitting argument.  Lemma~\ref{lem:unit} ensures that it is sufficient for us to analyze highest-density-first on instances where all jobs have unit size.  Notice that in this case, highest-density-first simply prioritizes jobs in order of largest weight.    Consider the following linear program, where $x_{i,t}$ is a  variable denoting how much $i$ is processed at time $t$ (in a true solution this will be either $0$ or $1$).

\begin{alignat*}{3}
\min & \quad  \sum_{i \in \mathcal S} \sum_{t \geq r_i} w_i (t-r_i) x_{i,t} \\
\text{s.t.} & \quad \sum_{t \geq r_i} x_{i,t} \geq 1 & \quad\forall i \in \mathcal S \\
& \quad \sum_{i \in \mathcal S : |\{u_i, v_i\} \cap \{w\}| = 1} x_{i,t} \leq \de_w  & \quad \forall w \in V,\ \forall t \in \N \\
& \quad x_{i,t} \geq 0 & \quad \forall i \in \mathcal S,\ \forall t \in \N
\end{alignat*}

As before we do not solve this LP, but rather use it only for analysis purposes. Note that the objective is the integral flow time.  The first set of constraints ensures each job is fully scheduled.  The second set of constraints ensures that the degree constraints are satisfied. 

\begin{lemma}
    If there is a schedule with weighted sum of flow times at most $F$, then there is a solution to the LP of cost at most $F$.
\end{lemma}
\begin{proof}
Consider a schedule $\{E_t\}_{t \in \mathbb{N}}$ with weighted sum of flow times $F$.  Since this is a feasible schedule, each $E_t$ satisfies the degree constraint at each vertex.  We create an LP solution as follows: if job $i$ is scheduled at time $t$  then we set $x_{i,t} = 1$, otherwise we set $x_{i,t} = 0$.  Since the original schedule is feasible, every job is scheduled at some point and thus the first LP constraint is satisfied.  Similarly, since each $E_t$ satisfies the degree constraints, the second set of  LP constraints are satisfied.  Thus this is a feasible LP solution.  The objective is the weighted flow time of the resulting schedule.
\end{proof}

The dual of this LP is the following. 
\begin{alignat*}{3}
\max & \quad  \sum_{i \in \mathcal S} \alpha_i - \sum_{u \in V} \sum_{t \in \N} \beta_{u,t} \\
\text{s.t.} & \quad \alpha_i - \frac{\beta_{u_i, t}}{d_{u_i}} - \frac{\beta_{v_i, t}}{d_{v_i}} \leq w_i(t - r_i)& \quad\forall i \in \mathcal S,\ \forall t \geq r_i  \\
& \quad \alpha_i \geq 0& \quad \forall u \in \mathcal S \\
& \quad \beta_{i,t} \geq 0& \quad \forall i \in \mathcal S,\ \forall t \in \N
\end{alignat*}

We will analyze our algorithm (highest-density-first, equivalent to highest-weight-first) by finding a feasible dual solution and relating this to the cost of the algorithm using resource augmentation.   Let $ALG(s)$ denote the total flow time of the algorithm when run with speedup $s$. 

Let's now define our dual solution.  But first we need a little bit of notation: for every node $k \in V$ and time $t$, let $\omega_k(t) = \sum_{i \in\mathcal S(t) : k \in \{u_i,v_i\}}  w_i$ denote the total  weight of jobs  adjacent to $k$ that have been released but are unsatisfied at time $t$.  Let $U_i(t)$ (resp.\ $V_i(t)$) be the jobs alive at time $t$ that share the end point $u_i$ (resp.\ $v_i$) with $i$. Then for every $i \in \mathcal S$, we set the $\alpha$ variables as follows. 
\begin{eqnarray*}
&&\alpha_i := \frac{1}{2s} \Bigg(\frac{1}{d_{u_i}} \left ( w_i \sum_{j \in U_i(r_i) : w_i <  w_j}   1 +   \sum_{j \in U_i(r_i) : w_i > w_j} w_j  \right )  \\
&&+ \frac{1}{d_{v_i}}  \left (w_i \sum_{j \in V_i(r_i) : w_i <  w_j}   1 +   \sum_{j \in V_i(r_i) : w_i > w_j} w_j \right )  \Bigg )
\end{eqnarray*}

It is not hard to see that, as in the simple setting of Section~\ref{sec:simple}, these dual variables correspond to an upper bound on the increase in the algorithm's cost due to the existence of job $i$.  Indeed, consider the first two terms depending on jobs $U_i$.  The first term states that job $i$ will wait on all jobs in $U_i(t)$ that have higher weight, and pay $w_i$ for each such time step.  The second term states that all lower weight jobs than $i$ will now need to wait on job $i$ before they are completed.  The last two terms are the same, but for jobs in $V_i(t)$.  Note that this expression is more complicated than in the simple setting since now we order by weights rather than by release time, so earlier jobs can be ``pushed back" due to job $i$ (unlike in the simple case).  

Similarly, we will set $\beta_{u,t} = \omega_u(t) / (2s)$, which is essentially the weighted version of the same dual variable in Section~\ref{sec:simple}.


We first show that this is a feasible dual solution.  Clearly all variables are nonnegative, so we just need to show the following lemma.

\begin{lemma} \label{lem:feasiblegeneral}
    $\alpha_i - \frac{\beta_{u_i, t}}{d_{u_i}} - \frac{\beta_{v_i, t}}{d_{v_i}} \leq w_i(t - r_i)$ for all $i \in \mathcal S$ and $t \geq r_i$.
\end{lemma}
 
%

\begin{proof} 
Consider any time $t \geq r_i$.  We have the following.
\begin{align}
 &   \alpha_i- \frac{\beta_{u_i, t}}{d_{u_i}} - \frac{\beta_{v_i, t}}{d_{v_i}} = \nonumber\\
 & \frac{1}{2s} \Bigg(\frac{1}{d_{u_i}} \bigg( w_i \sum_{j \in U_i(r_i), w_i <  w_j}   1 +   \sum_{j \in U_i(r_i), w_i > w_j} w_j  \bigg ) \nonumber \\
  &\;\;\;\; + \frac{1}{d_{v_i}}  \bigg (w_i \sum_{j \in V_i(r_i), w_i <  w_j}   1  +   \sum_{j \in V_i(r_i), w_i > w_j} w_j  \bigg )  \Bigg ) \nonumber \\  &\;\;\;\; -  \frac{\omega_{u_i}(t)}{2sd_{u_i}} -  \frac{\omega_{v_i}(t)}{2sd_{v_i}}    \label{fittingeqn}
\end{align}

We now bound the  first and third term by $\frac{1}{2}w_i(t-r_i)$, this is, half of the right hand side of the constraint.  The second and fourth will behave similarly.  Together, this will show the constraint is satisfied. 

We have the following.
\begin{eqnarray}
 && \frac{1}{2sd_{u_i}} \left ( w_i \sum_{j \in U_i(r_i), w_i <  w_j}   1 +   \sum_{j \in U_i(r_i), w_i > w_j} w_j \right )   \nonumber\\
  &&\;\;\;\; -   \frac{\omega_{u_i}(t)}{2sd_{u_i}} \nonumber \\
  &=&  \frac{1}{2sd_{u_i}} \left ( w_i \sum_{j \in U_i(r_i), w_i <  w_j}   1 +   \sum_{j \in U_i(r_i), w_i > w_j} w_j \right )  \nonumber \\
  &&\;\;\;\;   - \frac{1}{2sd_{u_i}}\sum_{j \in U_i(t)}   w_j  \;\;\;\; \mbox{[ def.\ of $\omega_{u_i}$]}\label{eq1}
\end{eqnarray}
  
Consider the last term.  Let $P_i(t) =U_i(r_i) \setminus U_i(t)$ denote the set of jobs in $U_i(r_i) $ that are completed (processed) by time $t$.  Then \eqref{eq1} is at most  the following, with equality if no jobs arrive during $[r_i,t]$. 
\begin{eqnarray*}
  &&\leq \frac{1}{2sd_{u_i}} \left ( w_i \sum_{j \in U_i(r_i), w_i <  w_j}   1 +   \sum_{j \in U_i(r_i), w_i > w_j} w_j \right )  \nonumber \\
  &&\;\;\;\; - \frac{1}{2sd_{u_i}}\sum_{j \in U_i(r_i) \setminus P_{i}(t)}   w_j  \label{eq2}
\end{eqnarray*}

Now we can use some of the jobs which appear in the the last term to cancel out the same jobs in the second term, and then use the relationship in the summations between the weights of jobs $j$ and $i$ to rewrite everything in terms of $w_i$.  This gives that~\eqref{eq2} is
\begin{small}
 \begin{eqnarray}
  &=& \frac{1}{2sd_{u_i}} \left ( w_i \sum_{j \in U_i(r_i), w_i <  w_j}   1  \right )   \nonumber \\
  &&\;\;\;\;- \frac{1}{2sd_{u_i}}\sum_{j \in U_i(r_i) \setminus P_{i}(t), w_i < w_j}   w_j  +  \frac{1}{2sd_{u_i}}\sum_{j \in  P_{i}(t), w_i > w_j}   w_j  \nonumber\\
    &\leq& \frac{1}{2sd_{u_i}} \left ( w_i \sum_{j \in U_i(r_i), w_i <  w_j}   1  \right ) \nonumber\\
  &&\;\;\;\;     - \frac{1}{2sd_{u_i}}\sum_{j \in U_i(r_i) \setminus P_{i}(t), w_i < w_j}   w_i \nonumber\\ 
  &&\;\;\;\; +  \frac{1}{2sd_{u_i}}\sum_{j \in  P_{i}(t), w_i > w_j}   w_i. \label{eq3}
\end{eqnarray}
\end{small}
 
 Now we combine the first term with the second to get that \eqref{eq3} is equal to
 \begin{eqnarray}
    &=& \frac{1}{2sd_{u_i}}\sum_{j \in  P_{i}(t), w_i < w_j}   w_i  +  \frac{1}{2sd_{u_i}}\sum_{j \in  P_{i}(t), w_i > w_j}   w_i \nonumber \\
    &=& \frac{w_i}{2sd_{u_i}} |P_i(t)|. \label{eq4}
\end{eqnarray}
 
We know that $\frac{1}{sd_{u_i}} |P_i(t)| \leq t-r_i$ because the algorithm can processes at most $s \cdot d_{u_i}$ jobs at each time step adjacent to $u_i$ and $ P_i(t)$ are jobs processed at $u_i$ during $[r_i,t]$.  Thus \eqref{eq5} is at most $\frac12 w_i(t-r_i)$.  Putting this all together, we have that 
\begin{eqnarray*}
&&\frac{1}{2sd_{u_i}} \left ( w_i \sum_{j \in U_i(r_i), w_i <  w_j}   1 +   \sum_{j \in U_i(r_i), w_i > w_j} w_j \right )   -   \frac{\omega_{u_i}(t)}{2sd_{u_i}}\\
  &&\leq \frac{1}{2} w_i  (t-r_i)  \label{eq5}   
\end{eqnarray*}
    
This bounds the first and third term of equation $\eqref{fittingeqn}$.  The second and fourth have the exact same analysis bounding them by $\frac{1}{2} w_i  (t-r_i) $.  Putting them together implies that $\eqref{fittingeqn}$ is bounded by $w_i(t-r_i)$, proving the lemma. 
\end{proof}

We will now prove two lemmas which will allow us to bound the cost of this dual solution.  Let $ALG(s)$ denote the total weighted flow time of the online algorithm. 

\begin{lemma} \label{lem:alphageneral}
    $\sum_{i \in \mathcal S} \alpha_i \geq \frac12 \cdot ALG(s)$.
\end{lemma}

\begin{proof}
Recall that $U_i(t)$ denotes all jobs that have not yet been processed by time $t$ which have $u_i$ as one endpoint (including job $i$ itself), and similarly for $V_i(t)$.  Then we have that
\begin{small}
\begin{align}
    \sum_{i \in \mathcal S}& (2\alpha_i) = \frac{1}{s} \sum_{i \in \mathcal S} \Bigg(\frac{1}{d_{u_i}} \left(w_i \sum_{j \in U_i(r_i) : w_i < w_j} 1 + \sum_{j \in U_i(r_i) : w_i > w_j} w_j\right) \nonumber \\
    & + \frac{1}{d_{v_i}} \left(w_i \sum_{j \in V_i(r_i) : w_i < w_j} 1 + \sum_{j \in V_i(r_i) : w_i > w_j} w_j\right) \Bigg) \nonumber\\
    &= \sum_{i \in \mathcal S} w_i \Bigg(\frac{1}{s d_{u_i}} |\{j \in U_i(r_i) : w_i < w_j\}| \nonumber \\
    & \qquad + \frac{1}{s d_{u_i}}|\{j: i \in U_i(r_j), w_i < w_j\}|) \nonumber \\
    &\qquad + \frac{1}{s d_{v_i}} |\{j \in V_i(r_i) : w_i < w_j\}| \nonumber \\
    &\qquad + \frac{1}{s d_{v_i}} |\{j : i \in V_i(r_j), w_i < w_j\}| \Bigg) \label{eq:upper}\\
    &\geq \sum_{i \in \mathcal S} w_i (C_i - r_i) = ALG(s) \nonumber.
\end{align}
\end{small}
The second equality  has arranged terms as follows.  Fix job $i$.  The first term counts jobs $j$ that require node $u_i$,  have higher weight than $i$, and are released and unsatisfied when $i$ arrives; this term comes from  $\alpha_i$.  The second term counts jobs $j$ with higher weight than $i$, that require node $u_i$, and  arrive during when $i$ is released at unsatisfied; this term comes from each such $\alpha_j$.  The last two terms are analogous for node $v_i$.

The final inequality is because the $i$'th term in the sum of \eqref{eq:upper} is an upper bound on the weighted flow time of job $i$.  This is because the only jobs which can prevent job $i$ from finished are either higher-weight jobs that show up earlier than $r_i$ at $u_i$ (the first term), higher-weight jobs which show up at $u_i$ after $r_i$ before job $i$ has finished (the second term), and similarly for jobs which show up at $v_i$ (the third and fourth terms).  Then we  multiply these jobs by the rate at which they are processed ($\frac{1}{s d_{u_i}}$ or $\frac{1}{s d_{v_i}}$).
\end{proof}

Next we bound the contribution of the $\beta$ variables. 

\begin{lemma} \label{lem:betageneral}
    $\sum_{w \in V} \sum_{t \in \N} \beta_{w,t} \leq \frac{1}{s} \cdot ALG(s)$.
\end{lemma}
\begin{proof}
This is essentially a straightforward calculation using the definition of weighted flow time and the fact that each job has two endpoints. Let $c(i)$  be the completion time of $i$ in highest-density-first's schedule.  We have the following.
\begin{align*}
    \sum_{a \in V} \sum_{t \in \N} \beta_{a,t} &= \frac{1}{2s} \sum_{t \in \N} \sum_{a \in V} \sum_{i \in \mathcal S(t) : a \in \{u_i,v_i\}} w_i \\
    &= \frac{1}{s} \sum_{i \in \mathcal S}\sum_{r_i \leq t \leq c(i)} w_i = \frac{1}{s} \cdot ALG(s). \qedhere
\end{align*} \end{proof}

We can now prove our main theorem.  In the following, let $OPT$ be the optimal solution (without speedup).

\begin{lemma}\label{lem:mainunit}
$ALG(2+\epsilon) \leq \frac{2\eps+4}{\eps}\cdot OPT$ for any $\epsilon > 0$.
\end{lemma}
\begin{proof}
Let $s = 2+\epsilon$.  Combining Lemmas~\ref{lem:alphageneral} and \ref{lem:betageneral} implies that 
\begin{small}
\begin{align*}
    &\sum_{i \in \mathcal S} \alpha_i - \sum_{w \in V} \sum_{t \in \N} \beta_{w,t} \\
    &\geq \frac12 \cdot ALG(2+\epsilon)  - \frac{1}{2+\epsilon} \cdot ALG(2+\epsilon) \geq \frac{\eps}{2\eps+4} \cdot ALG(2+\epsilon).
\end{align*}
\end{small}
We know from Lemma~\ref{lem:feasiblegeneral} that $(\alpha, \beta)$ is a feasible dual solution, so by weak duality we get that
\begin{align*}
    &ALG(2+\epsilon) \\
    &\leq \frac{2\eps+4}{\eps}\cdot \left(\sum_{i \in \mathcal S} \alpha_i - \sum_{w \in V} \sum_{t \in \N} \beta_{w,t} \right) \leq \frac{2\eps+4}{\eps} \cdot OPT. \qedhere
\end{align*}
\end{proof}

Finally we get our main theorem by combining the previous lemma with the reduction to the unit time instance in Lemma~\ref{lem:unit}.  Note that by setting $\epsilon$ to any appropriate constant (say, $1/2$), Theorem~\ref{thm:mainarb} gives an $O(1)$-competitive algorithm with $O(1)$-speedup.

\begin{theorem}\label{thm:mainarb}
Highest-density-first is $(2+\eps)$-speed $O(\frac{1}{\eps^2})$-competitive for the total flow time objective when jobs have arbitrary sizes and weights and the degree bounds are arbitrary for any $0 < \epsilon \leq 1$. 
\end{theorem}

\subsection{Completion Times} \label{sec:completion}
We now claim that Theorem~\ref{thm:mainarb} implies there is a $O(1)$-competitive for the total (weighted) completion time objective function, even without any speedup.  To see this, we argue that we can simulate speed-up $s$ for the total completion time objective by losing a factor $s$ in the competitive ratio.  Given any online schedule $A$ using $s$-speed, construct a online schedule $B$ using $1$-speed as follows.  Each job scheduled with $s$ speed at time $t$ in $A$ is scheduled during the interval $(st, s(t+1)]$ in $B$.  This ensures a job $i$ completed at time $c(i)$ in $A$ is completed at time $s \cdot c(i)$ in $B$.  Thus, each job pays an extra factor of at most $s$ in the completion time, so this extra factor goes directly into the competitive ratio.   

More formally, we prove the following (where we make no attempt to optimize the constant).
\begin{theorem} \label{thm:completion}
There is a  $O(1)$-competitive for the total completion time objective when jobs have arbitrary sizes and weights and the degree bounds are arbitrary.
\end{theorem}
\begin{proof}
Let $OPT_c$ denote the cost of the optimal schedule with respect to weighted completion times, and let $c^*(i)$ denote the completion time of job $i$ in this schedule.  Note that the total weighted flow time of this schedule is $F = \sum_i w_i(c^*(i) - r_i) = \sum_i w_i c^*(i) - \sum_i w_i r_i$.


Let $c_f(i)$ denote the completion time of job $i$ when we run highest-density-first with $3$-speed.   Then Theorem~\ref{thm:mainarb} implies that $\sum_i w_i(c_f(i) - r_i) \leq O(1) \cdot F = O(1) \cdot (OPT_c - \sum_i w_i r_i)$.  Now by stretching out time as described earlier, we get a new schedule where job $i$ completes at time at most $c(i) \leq 3 \cdot c_f(i)$.  Putting this together, we get that

\begin{small}
\begin{align*}
    &\sum_i w_i c(i) \leq 3 \sum_i w_i  c_f(i) \\
    &= 3 \left(\sum_i w_i  c_f(i) - \sum_i w_i r_i\right) + 3 \sum_i w_i r_i \\
    &\leq 3\left( O(1) \cdot \left(OPT_c - \sum_i w_i r_i\right)\right) + 3 \sum_i w_i r_i\\
    &\leq O(1) \cdot OPT_c. \qedhere
\end{align*}
\end{small}
\end{proof}

Despite the wide variety of algorithms and analyses they provided for the sum of completion times, the existence of such an algorithm for this general setting \emph{was not} given in~\cite{JJGDG17}.  They did not give bounds on \emph{weighted} completion times in any setting, and even for unweighted completion times they did not provide a $O(1)$-competitive algorithm for arbitrary sizes, degree bounds, and release times.  Thus this shows that designing algorithms for flow times, even with speedup, can yield improvements for  completion times.

\section{Lower Bound} \label{sec:lower-bound}

In this section we will prove the following theorem.
\begin{theorem} \label{thm:LB-randomize}
Every randomized algorithm has expected competitive ratio at least $\Omega(\sqrt{n})$, where $n$ is the total number of jobs, even on instances in which all job sizes are $1$ and all degrees are $1$.
\end{theorem}

To prove this, we first apply Yao's principle~\cite{Yao}: it is sufficient to provide a distribution over inputs such that any deterministic algorithm has expected competitive ratio of at least $\Omega(\sqrt{n})$.  So consider the following distribution.  

Let $V = \{v_1, v_2, v_3, v_4\}$, and set all degree bounds to $1$.  Let $L$ be some large even value (eventually $n$ will be $\Theta(L)$).  Let $\mathcal S_1$ be a set of $\sqrt{L}$ identical jobs, each of the form $(v_1, v_2, 1, 1)$, and similarly let $\mathcal S_2$ be a set of $\sqrt{L}$ identical jobs each of the form $(v_3, v_2, 1, 1)$.  (Note that both of these sets consist of jobs which are released at time $1$).  Let $\mathcal S_3 = \{(v_3, v_4, 1, \sqrt{L} + i) :  i \in [L]\}$ (so one job released at each time in $[\sqrt{L} + 1, L + \sqrt{L}]$), and let $\mathcal S_4 = \{(v_1, v_4, 1, \sqrt{L} + i) :  i \in [L]\}$.  Our distribution of instances is the following: with probability $1/2$ the set of jobs is $\mathcal T_1 = \mathcal S_1 \cup \mathcal S_2 \cup \mathcal S_3$, and with probability $1/2$ the set of jobs is $\mathcal T_2 = \mathcal S_1 \cup \mathcal S_2 \cup \mathcal S_4$.  Note that in both cases, $n = L + 2\sqrt{L} = \Theta(L)$.

\begin{lemma} \label{lem:LB-OPT}
$OPT \leq O(n)$ with probability $1$
\end{lemma}
\begin{proof}
If the actual instance is $\mathcal T_1$, then for every $t \in [\sqrt{L}]$, OPT could schedule a job in $\mathcal S_2$ (since they are all released at time $1$).  Then after time $\sqrt{L}$, all jobs from $\mathcal S_2$ have been completed.  Then for the next $\sqrt{L}$ rounds, OPT can schedule one job from $\mathcal S_1$ and one job from $\mathcal S_3$ simultaneously (since they do not share any endpoints, and one new job from $\mathcal S_3$ arrives in each round).  Then after round $2\sqrt{L}$ all jobs in $\mathcal S_1$ have been completed, so OPT will continue to schedule the jobs in $\mathcal S_3$ as they arrive.  In this schedule, every job in $\mathcal S_1$ has flow time at most $\sqrt{L}$, every job in $\mathcal S_2$ has flow time at most $2\sqrt{L}$, and every job in $\mathcal S_3$ has flow time $1$.  Thus $OPT \leq \sqrt{L} \cdot \sqrt{L} + \sqrt{L} \cdot 2\sqrt{L} + L = O(L)$.

Similarly, if the actual instance is $\mathcal T_2$, then for every $t \in [\sqrt{L}]$, OPT could schedule a job in $\mathcal S_1$.  Then after time $\sqrt{L}$, all jobs from $\mathcal S_1$ have been completed.  Then for the next $\sqrt{L}$ rounds, OPT can schedule one job from $\mathcal S_2$ and one job from $\mathcal S_4$ simultaneously.  Then after round $2\sqrt{L}$ all jobs in $\mathcal S_2$ have been completed, so OPT will continue to schedule the jobs in $\mathcal S_4$ as they arrive.  As in the $\mathcal T_1$ case, the total flow time achieved by OPT is at most $O(L) = O(n)$.
\end{proof}

Now we analyze an arbitrary deterministic online algorithm $\mathcal A$.  We begin with the following claim.
\begin{lemma} \label{lem:LB-ALG1}
With probability at least $1/2$, for all $t \in \{\sqrt{L}+1, \sqrt{L} + 2, \dots, L + \sqrt{L}\}$, there are at least $\sqrt{L}/2$ unfinished jobs at time $t$ that have already been released.
\end{lemma}
\begin{proof}
Both of the possible instances are the same up until time $\sqrt{L}$, and by time $\sqrt{L}$, $\mathcal A$ has completed at most $\sqrt{L}$ jobs from $\mathcal S_1 \cup S_2$ (since they all share at least one endpoint).  This after time $\sqrt{L}$, either $\mathcal S_1$ or $\mathcal S_2$ still has at least $\sqrt{L} / 2$ unfinished jobs.  

If $\mathcal S_2$ still has at least $\sqrt{L} / 2$ unfinished jobs (case 1), then suppose that the instance is $\mathcal T_1$ (this happens with probability $1/2$).  We prove the lemma by induction on $t$.  When $t = \sqrt{L}+1$, we know that there are at least $\sqrt{L}/2$ jobs from $\mathcal S_2$ that have not yet been completed.  So the lemma is true for $t = \sqrt{L} + 1$.  Now consider some $\sqrt{L} + 1 < t \leq L + \sqrt{L}$.  By induction, at time $t-1$ there were at least $\sqrt{L}/2$ uncompleted jobs from $\mathcal S_2 \cup \mathcal S_3$ that had already been released.  At most one of these jobs was processed by $\mathcal A$ at time $t-1$ (since they all share $v_3$ as an endpoint), and at time $t$ one new job from $\mathcal S_3$ was released.  Thus the number of uncompleted jobs from $\mathcal S_2 \cup S_3$ at time $t$ is at least $\sqrt{L}/2 - 1  + 1 = \sqrt{L}$, as claimed.  

Now suppose that $\mathcal S_1$ still has at least $\sqrt{L} / 2$ unfinished jobs after time $\sqrt{L}$ (case 2).  Then with probability $1/2$ the instance is $\mathcal T_2$.  The same induction works here.  When $t = \sqrt{L}+1$, we know that there are at least $\sqrt{L}/2$ jobs from $\mathcal S_1$ that have not yet been completed, so the lemma is true for $t = \sqrt{L} + 1$.  Now consider some $\sqrt{L} + 1 < t \leq L + \sqrt{L}$.  By induction, at time $t-1$ there were at least $\sqrt{L}/2$ uncompleted jobs from $\mathcal S_1 \cup \mathcal S_4$ that had already been released.   At most one of these jobs was processed by $\mathcal A$ at time $t-1$ (since they all share $v_1$ as an endpoint), and at time $t$ one new job from $\mathcal S_4$ was released.  Thus the number of uncompleted jobs from $\mathcal S_1 \cup S_4$ at time $t$ is at least $\sqrt{L}/2 - 1  + 1 = \sqrt{L}$, as claimed. \end{proof}
\begin{lemma} \label{lem:LB-A}
The expected sum of flow times in $\mathcal A$ is at least $\Omega(n^{3/2})$
\end{lemma}
\begin{proof} For every job $i$, let $c(i)$ denote its completion time in $\mathcal A$.  For every time, let $R(t)$ denote the number of jobs that have been released but not yet completed by $\mathcal A$.  Let $\mathcal S = \mathcal T_1$ if $T_1$ is the instance, and otherwise let $\mathcal S = \mathcal T_2$.  Then Lemma~\ref{lem:LB-ALG1} implies that with probability at least $1/2$,
\begin{small}
\begin{align*}
    \sum_{i \in \mathcal S} (c(i) - r_i) &= \sum_{t} R(t) \geq \sum_{t=\sqrt{L}+1}^{L + \sqrt{L}} R(t) \geq \sum_{t=\sqrt{L}+1}^{L + \sqrt{L}} \frac{\sqrt{L}}{2} \\
    &\geq \Omega(L^{3/2}) = \Omega(n^{3/2}). \qedhere
\end{align*}
\end{small}
\end{proof}

Lemmas~\ref{lem:LB-OPT} and~\ref{lem:LB-A}, together with Yao's principle~\cite{Yao}, imply Theorem~\ref{thm:LB-randomize}.


\bibliographystyle{plain}
\bibliography{refs}

\begin{thebibliography}{10}

\bibitem{AnandGK12}
S.~Anand, Naveen Garg, and Amit Kumar.
\newblock Resource augmentation for weighted flow-time explained by dual
  fitting.
\newblock In {\em Proceedings of the Twenty-Third Annual {ACM-SIAM} Symposium
  on Discrete Algorithms, {SODA} 2012, Kyoto, Japan, January 17-19, 2012},
  pages 1228--1241, 2012.

\bibitem{DAN3}
Chen Avin, Alexandr Hercules, Andreas Loukas, and Stefan Schmid.
\newblock \emph{rDAN}: Toward robust demand-aware network designs.
\newblock {\em Inf. Process. Lett.}, 133:5--9, 2018.

\bibitem{DAN4}
Chen Avin, Kaushik Mondal, and Stefan Schmid.
\newblock Demand-aware network designs of bounded degree.
\newblock In Andr{\'{e}}a~W. Richa, editor, {\em 31st International Symposium
  on Distributed Computing, {DISC} 2017, October 16-20, 2017, Vienna, Austria},
  volume~91 of {\em LIPIcs}, pages 5:1--5:16. Schloss Dagstuhl -
  Leibniz-Zentrum fuer Informatik, 2017.

\bibitem{DAN2}
Chen Avin, Kaushik Mondal, and Stefan Schmid.
\newblock Demand-aware network design with minimal congestion and route
  lengths.
\newblock In {\em 2019 {IEEE} Conference on Computer Communications, {INFOCOM}
  2019, Paris, France, April 29 - May 2, 2019}, pages 1351--1359. {IEEE}, 2019.

\bibitem{DAN1}
Chen Avin and Stefan Schmid.
\newblock Toward demand-aware networking: A theory for self-adjusting networks.
\newblock {\em SIGCOMM Comput. Commun. Rev.}, 48(5):31--40, January 2019.

\bibitem{BansalC09}
Nikhil Bansal and Ho{-}Leung Chan.
\newblock Weighted flow time does not admit o(1)-competitive algorithms.
\newblock In Claire Mathieu, editor, {\em Proceedings of the Twentieth Annual
  {ACM-SIAM} Symposium on Discrete Algorithms, {SODA} 2009, New York, NY, USA,
  January 4-6, 2009}, pages 1238--1244. {SIAM}, 2009.

\bibitem{BecchettiLMP06}
Luca Becchetti, Stefano Leonardi, Alberto Marchetti{-}Spaccamela, and Kirk
  Pruhs.
\newblock Online weighted flow time and deadline scheduling.
\newblock {\em J. Discrete Algorithms}, 4(3):339--352, 2006.

\bibitem{reconfigurable-tutorial}
Ramakrishnan Durairajan, Klaus-Tycho Foerster, and Stefan Schmid.
\newblock Reconfigurable networks: Enablers, algorithms, complexity (renets),
  June 2019.
\newblock Tutorial at SIGMETRICS 2019.

\bibitem{helios}
Nathan Farrington, George Porter, Sivasankar Radhakrishnan, Hamid~Hajabdolali
  Bazzaz, Vikram Subramanya, Yeshaiahu Fainman, George Papen, and Amin Vahdat.
\newblock Helios: A hybrid electrical/optical switch architecture for modular
  data centers.
\newblock In {\em Proceedings of the ACM SIGCOMM 2010 Conference}, SIGCOMM '10,
  pages 339--350, New York, NY, USA, 2010. ACM.

\bibitem{FGS18}
Klaus-Tycho Foerster, Manya Ghobadi, and Stefan Schmid.
\newblock Characterizing the algorithmic complexity of reconfigurable data
  center architectures.
\newblock In {\em Proceedings of the 2018 Symposium on Architectures for
  Networking and Communications Systems}, ANCS '18, pages 89--96, New York, NY,
  USA, 2018. ACM.

\bibitem{FPS19}
Klaus-Tycho Foerster, Maciej Pacut, and Stefan Schmid.
\newblock On the complexity of non-segregated routing in reconfigurable data
  center architectures.
\newblock {\em SIGCOMM Comput. Commun. Rev.}, 49(2):2--8, May 2019.

\bibitem{survey}
Klaus-Tycho Foerster and Stefan Schmid.
\newblock Survey of reconfigurable data center networks: Enablers, algorithms,
  complexity.
\newblock {\em SIGACT News}, 50(2):62--79, July 2019.

\bibitem{projector}
Monia Ghobadi, Ratul Mahajan, Amar Phanishayee, Nikhil Devanur, Janardhan
  Kulkarni, Gireeja Ranade, Pierre-Alexandre Blanche, Houman Rastegarfar,
  Madeleine Glick, and Daniel Kilper.
\newblock Projector: Agile reconfigurable data center interconnect.
\newblock In {\em Proceedings of the 2016 ACM SIGCOMM Conference}, SIGCOMM '16,
  pages 216--229, New York, NY, USA, 2016. ACM.

\bibitem{firefly}
Navid Hamedazimi, Zafar Qazi, Himanshu Gupta, Vyas Sekar, Samir~R. Das, Jon~P.
  Longtin, Himanshu Shah, and Ashish Tanwer.
\newblock Firefly: A reconfigurable wireless data center fabric using
  free-space optics.
\newblock In {\em Proceedings of the 2014 ACM Conference on SIGCOMM}, SIGCOMM
  '14, pages 319--330, New York, NY, USA, 2014. ACM.

\bibitem{ImMP11}
Sungjin Im, Benjamin Moseley, and Kirk Pruhs.
\newblock A tutorial on amortized local competitiveness in online scheduling.
\newblock {\em {SIGACT} News}, 42(2):83--97, 2011.

\bibitem{JJGDG17}
Su~Jia, Xin Jin, Golnaz Ghasemiesfeh, Jiaxin Ding, and Jie Gao.
\newblock Competitive analysis for online scheduling in software-defined
  optical {WAN}.
\newblock In {\em 2017 {IEEE} Conference on Computer Communications, {INFOCOM}
  2017, Atlanta, GA, USA, May 1-4, 2017}, pages 1--9. {IEEE}, 2017.

\bibitem{owan}
Xin Jin, Yiran Li, Da~Wei, Siming Li, Jie Gao, Lei Xu, Guangzhi Li, Wei Xu, and
  Jennifer Rexford.
\newblock Optimizing bulk transfers with software-defined optical wan.
\newblock In {\em Proceedings of the 2016 ACM SIGCOMM Conference}, SIGCOMM '16,
  pages 87--100, New York, NY, USA, 2016. ACM.

\bibitem{KalyanasundaramP00}
Bala Kalyanasundaram and Kirk Pruhs.
\newblock Speed is as powerful as clairvoyance.
\newblock {\em J. {ACM}}, 47(4):617--643, 2000.

\bibitem{flyways}
Srikanth Kandula, Jitendra Padhye, and Paramvir Bahl.
\newblock Flyways to de-congest data center networks.
\newblock In Lakshminarayanan Subramanian, Will~E. Leland, and Ratul Mahajan,
  editors, {\em Eight {ACM} Workshop on Hot Topics in Networks (HotNets-VIII),
  {HOTNETS} '09, New York City, NY, USA, October 22-23, 2009}. {ACM} {SIGCOMM},
  2009.

\bibitem{PruhsST04}
Kirk Pruhs, Jir{\'{\i}} Sgall, and Eric Torng.
\newblock Online scheduling.
\newblock In Joseph~Y.{-}T. Leung, editor, {\em Handbook of Scheduling -
  Algorithms, Models, and Performance Analysis.} Chapman and Hall/CRC, 2004.

\bibitem{cThrough}
Guohui Wang, David~G. Andersen, Michael Kaminsky, Konstantina Papagiannaki,
  T.S.~Eugene Ng, Michael Kozuch, and Michael Ryan.
\newblock c-through: Part-time optics in data centers.
\newblock In {\em Proceedings of the ACM SIGCOMM 2010 Conference}, SIGCOMM '10,
  pages 327--338, New York, NY, USA, 2010. ACM.

\bibitem{Yao}
A.~C. {Yao}.
\newblock Probabilistic computations: Toward a unified measure of complexity.
\newblock In {\em 18th Annual Symposium on Foundations of Computer Science
  (sfcs 1977)}, pages 222--227, Oct 1977.

\bibitem{mirrormirror}
Xia Zhou, Zengbin Zhang, Yibo Zhu, Yubo Li, Saipriya Kumar, Amin Vahdat, Ben~Y.
  Zhao, and Haitao Zheng.
\newblock Mirror mirror on the ceiling: Flexible wireless links for data
  centers.
\newblock In {\em Proceedings of the ACM SIGCOMM 2012 Conference on
  Applications, Technologies, Architectures, and Protocols for Computer
  Communication}, SIGCOMM '12, pages 443--454, New York, NY, USA, 2012. ACM.

\end{thebibliography}

\end{document}